\documentclass[a4paper]{article}
\usepackage{graphicx} 
\usepackage{amsmath, amsthm, amsfonts, amssymb, amscd, bm}
\usepackage{verbatim}
\usepackage{enumerate}
\usepackage{url}
\usepackage{pgf,tikz}
\usetikzlibrary{arrows}

\newtheorem{theorem}{Theorem}[section]
\newtheorem{corollary}{Corollary}[theorem]
\newtheorem{lemma}{Lemma}[section]
\newtheorem{definition}{Definition}[section]
\newtheorem{proposition}{Proposition}[section]

\theoremstyle{remark}
\newtheorem{remark}{Remark}[section]

\DeclareMathOperator{\Pesssup}{\mathcal{P}-\,ess\,sup}

\definecolor{xdxdff}{rgb}{0.65882,0.65882,0.65882}
\definecolor{zzttqq}{rgb}{0.26667,0.26667,0.26667}
\definecolor{uququq}{rgb}{0.25098,0.25098,0.25098}
\definecolor{qqqqff}{rgb}{0.33333,0.33333,0.33333}

\begin{document}

\title{What risk measures are time consistent for all filtrations?}
\author{Samuel N. Cohen\\University of Adelaide\\samuel.cohen@adelaide.edu.au}

\maketitle
\begin{abstract}
We study coherent risk measures which are time-consistent for multiple filtrations. We show that a coherent risk measure is time-consistent for every filtration if and only if it is one of four main types. Furthermore, if the risk measure is strictly monotone it is linear, and if the reference probability space is not atomic then it is either linear or an essential supremum.
\end{abstract}

\section{Introduction}
The theory of coherent risk measures has received much attention in recent research. For a probability space $(\Omega, \mathcal{F},\mathbb{P}^0)$, where $\mathbb{P}^0$ is a reference probability measure, these are maps $\rho:L^\infty(\mathcal{F})\to\mathbb{R}$ satisfying assumptions of monotonicity, sublinearity, positive homogeneity, cash invariance and, typically, a Fatou property. An extensive discussion of these maps and their extensions can be found in \cite{Artzner1999}, \cite{Delbaen2002}, or \cite{Follmer2002}, amongst others.

One area of particular activity is the study of time-consistent risk measures, that is, where, for a given filtration $\{\mathcal{F}_t\}$ (we assume that $\mathcal{F}_0$ is trivial), one can find a family of maps $\rho_t:L^\infty (\mathcal{F})\to L^\infty(\mathcal{F}_t)$ such that $\rho_0\equiv \rho$, and a semigroup property $\rho_s(-\rho_t(X))=\rho_s(X)$ is satisfied for all $X$. In this note, we seek conditions such that a time consistent extension exists for all filtrations.

In \cite{Artzner1999} and \cite{Delbaen2002}, it is shown that a coherent risk measure with the Fatou property has the following representation.
\[\rho(X) = \max_{\mathbb{P}\in\mathcal{P}}\{E_{\mathbb{P}}[-X]\}\]
where $\mathcal{P}$ is a $L^1(\mathbb{P}^0)$-closed convex family of probability measures absolutely continuous with respect to the reference measure $\mathbb{P}^0$. Given this representation, we can consider these risk measures through their Radon-Nikodym derivatives $Z:=d\mathbb{P}/d\mathbb{P}^0$. We denote by $\mathcal{Z}$ the set of random variables obtained from $\mathcal{P}$ using this method.

For a classical (linear) expectation, a classical result establishes the existence of the conditional expectations $E[\cdot|\mathcal{F}_t]$ for any filtration $\{\mathcal{F}_t\}$. The induced dynamic risk measure $\rho_t(X) = -E[X|\mathcal{F}_t]$ is then time consistent for the filtration $\{\mathcal{F}_t\}$. We shall show that the existence of such an extension for every filtration implies that the risk measure must be one of four special types.

Such a result is of practical interest, as it raises issues of robustness of time-consistent nonlinear risk measures to errors in the modelling of the filtration. In particular, our result shows that, for every coherent risk measure which is not of these types, there will be at least one (simple) filtration with respect to which it is not time-consistent.

\section{A representation result}

\begin{theorem}[Pasting Property, Artzner et al. \cite{Artzner2007}]
For a given filtration $\{\mathcal{F}_t\}$, a coherent risk measure admits a time-consistent extension with respect to $\{\mathcal{F}_t\}$, (or, more simply, is $\{\mathcal{F}_t\}$-consistent) if and only if we have the `pasting property'
\[Z, Z'\in\mathcal{Z} \text{ implies }\quad \left( Z \cdot \frac{E[Z'|\mathcal{F}_t]}{E[Z|\mathcal{F}_t]}\right) \in\mathcal{Z} \text{ for all }t.\]
\end{theorem}

This property is the primary tool we shall use to obtain our result.

\begin{definition}
We shall say that a filtration is simple if it is of the form 
\[\mathcal{F}_s = \begin{cases}
 \{\emptyset, \Omega\} &s< t\\
 \{\emptyset, A, A^c, \Omega\}=\sigma(A)& t\leq s < t'\\
 \mathcal{F} & t' \leq s
 \end{cases}\]
 for some times $t, t'\in ]0,T]$ and some measurable set $A\in\mathcal{F}$. We denote such a filtration $\mathcal{F}^A$.
 \end{definition}
 
\begin{definition}
If a set $A\in\mathcal{F}$ has $\mathbb{P}(A)=0$ for all $\mathbb{P}\in\mathcal{P}$, then it is called \emph{polar}. For a random variable $X$, $\Pesssup(X)$ is the least constant $c$ such that $\{\omega:c<X(\omega)\}$ is polar.
\end{definition}

\begin{lemma} \label{lem:mainlem}
Let $\rho:L^\infty\to\mathbb{R}$ be a coherent risk measure with the Fatou property, and assume $\rho$ admits a time-consistent extension with respect to two simple filtrations $\mathcal{F}^A,\mathcal{F}^B$, where $A\cap B = \emptyset$. Let $\mathcal{P}$ be a representative closed set of test measures for $\rho$. Then at least one of the following holds.
\begin{enumerate}[(i)]
\item At least one of $A, B, (A\cup B)^c$ is polar.
\item $\max_{\mathbb{P}\in\mathcal{P}} \mathbb{P}(A)=1$ and $\min _{\mathbb{P}\in\mathcal{P}} \mathbb{P}(A)=0$, and similarly for $B$ and $(A\cup B)^c$
\item $\mathbb{P}(A), \mathbb{P}(B), \mathbb{P}(B|A^c)$ and $\mathbb{P}(A|B^c)$ are fixed for $\mathbb{P}\in\mathcal{P}$.
\end{enumerate}
\end{lemma}

\begin{proof}
If we are in cases (i) or (iii) we are done.

As $A$ and $B$ are disjoint, it is easy to show
\[\mathbb{P}(A) = \frac{\mathbb{P}(A|B^c)[1-\mathbb{P}(B|A^c)]}{1-\mathbb{P}(A|B^c)\mathbb{P}(B|A^c)},\]
and similarly for $B$. Consequently, as we are not in case (iii), at least one of $\mathbb{P}(B|A^c)$ and $\mathbb{P}(A|B^c)$ must vary. Without loss of generality, we shall assume $\mathbb{P}(B|A^c)$ may vary.

Now suppose that $\mathbb{P}(A)>0$ for all $\mathbb{P}\in\mathcal{P}$. It follows that $\mathbb{P}(B^c)>0, \mathbb{P}((A\cup B)^c)<1$ for all $\mathbb{P}\in\mathcal{P}$.  Let $Z^1\in \mathcal{Z}$ correspond to a measure $\mathbb{P}^1$ which minimises $\mathbb{P}^1(A)$, this exists by the closedness of $\mathcal{P}$. Let $Z^2$, $Z^3$ correspond to measures with $\mathbb{P}^2(B|A^c)<\mathbb{P}^3(B|A^c)$. Hence define, 
\[Z^4 = \frac{E[Z^1|\sigma(A)]}{E[Z^2|\sigma(A)]} Z^2.\]
Note that 
\[\mathbb{P}^4(A)=E[I_A Z^4] = E[I_A Z^1]=\mathbb{P}^1(A)\]
and 
\[\mathbb{P}^4(B|A^c)=\frac{E[I_B Z^4|A^c]}{E[Z^4|A^c]} =\frac{E[I_B Z^2|A^c]}{E[Z^2|A^c]}=\mathbb{P}^2(B|A^c).\]
Similarly define 
\[Z^5 = \frac{E[Z^1|\sigma(A)]}{E[Z^3|\sigma(A)]} Z^3.\]
As $\rho$ is time-consistent for $\mathcal{F}^A$, we know that $Z^4, Z^5\in\mathcal{Z}$. By construction, we also have
$\mathbb{P}^4(B|A^c)< \mathbb{P}^5(B|A^c)$ and $\mathbb{P}^4(A)=\mathbb{P}^5(A)$, hence $\mathbb{P}^4(B)< \mathbb{P}^5(B)$.

Finally, let
\[Z^6 = \frac{E[Z^5|\sigma(B)]}{E[Z^4|\sigma(B)]} Z^4 \in\mathcal{Z}.\]

It follows that $\mathbb{P}^6(B) = \mathbb{P}^5(B) > \mathbb{P}^4(B)$, and that 
\[\mathbb{P}^6(A|B^c) = \frac{E[I_A Z^6|B^c]}{E[Z^6|B^c]} = \frac{E[I_A Z^4|B^c]}{E[Z^4|B^c]} = \mathbb{P}^4(A|B^c).\]
Therefore, 
\[\mathbb{P}^6(A) = \mathbb{P}^6(A|B^c)\cdot \mathbb{P}^6(B^c) < \mathbb{P}^4(A|B^c)\cdot\mathbb{P}^4(B^c) = \mathbb{P}^4(A) = \mathbb{P}^1(A) \]
which contradicts the minimality of $\mathbb{P}^1(A)$. Therefore $\min_{\mathbb{P}\in\mathcal{P}} \mathbb{P}(A)=0$. 

In a similar way we can show that $\max_{\mathbb{P}\in\mathcal{P}} \mathbb{P}(A)=1$, and the required results for $\mathbb{P}(B)$. The construction is apparent from the geometrical argument below, the details are left to the reader.
\end{proof}

\begin{remark}
An graphical representation of this result is the following modification of an argument of Epstein and Schneider \cite{Epstein2003}.  In \cite{Epstein2003} the state space is finite. We do not make this assumption, however, as we are considering simple filtrations, the interim $\sigma$-algebras are finite, and their method can be appropriately adapted.

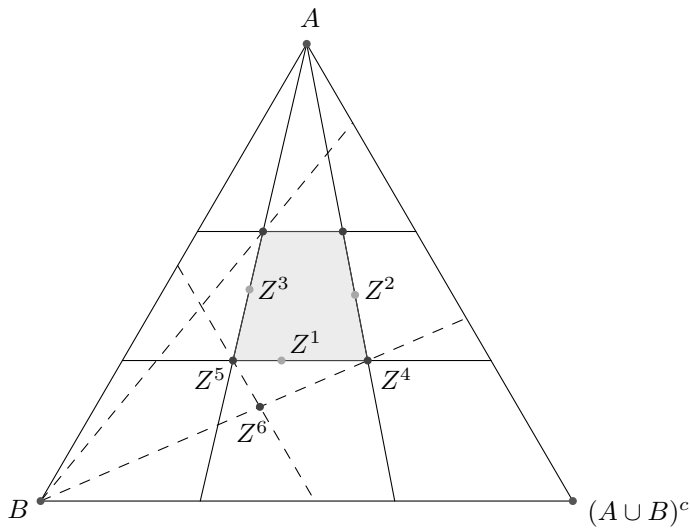
\begin{figure} \label{fig1}
\begin{tikzpicture}[line cap=round,line join=round,>=triangle 45,x=0.7cm,y=0.7cm]
\clip(-8.30071,-2.86262) rectangle (8.13123,10.34737);
\fill[color=zzttqq,fill=zzttqq,fill opacity=0.1] (-0.82081,5.10603) -- (0.67889,5.10603) -- (1.1458,2.66161) -- (-1.38533,2.66161) -- cycle;
\draw (0,8.66025)-- (5,0);
\draw (5,0)-- (-5,0);
\draw (-5,0)-- (0,8.66025);
\draw (0,8.66025)-- (-2,0);
\draw (0,8.66025)-- (1.65419,0);
\draw (-2.05203,5.10603)-- (2.05203,5.10603);
\draw (-3.46332,2.66161)-- (3.46332,2.66161);
\draw [color=zzttqq] (-0.82081,5.10603)-- (0.67889,5.10603);
\draw [color=zzttqq] (0.67889,5.10603)-- (1.1458,2.66161);
\draw [color=zzttqq] (1.1458,2.66161)-- (-1.38533,2.66161);
\draw [color=zzttqq] (-1.38533,2.66161)-- (-0.82081,5.10603);
\draw [dash pattern=on 4pt off 4pt] (-5,0)-- (0.86375,7.16419);
\draw [dash pattern=on 4pt off 4pt] (-5,0)-- (2.99976,3.46452);
\draw [dash pattern=on 4pt off 4pt] (-2.42432,4.4612)-- (0.15135,0);
\draw (-0.5002,3.4) node[anchor=north west] {$Z^1$};
\draw (0.9,4.4122) node[anchor=north west] {$Z^2$};
\draw (-1.1,4.4122) node[anchor=north west] {$Z^3$};
\draw (1.2,2.7) node[anchor=north west] {$Z^4$};
\draw (-2.3,2.7) node[anchor=north west] {$Z^5$};
\draw (-1.5,1.7) node[anchor=north west] {$Z^6$};
\draw (-0.3,9.5) node[anchor=north west] {$A$};
\draw (-5.8,0.2) node[anchor=north west] {$B$};
\draw (5.1,0.2) node[anchor=north west] {$(A\cup B)^c$};
\fill [color=qqqqff] (0,8.66025) circle (1.5pt);
\fill [color=qqqqff] (-5,0) circle (1.5pt);
\fill [color=qqqqff] (5,0) circle (1.5pt);
\fill [color=uququq] (-0.82081,5.10603) circle (1.5pt);
\fill [color=uququq] (0.67889,5.10603) circle (1.5pt);
\fill [color=uququq] (1.1458,2.66161) circle (1.5pt);
\fill [color=uququq] (-1.38533,2.66161) circle (1.5pt);
\fill [color=uququq] (-0.87904,1.7847) circle (1.5pt);
\fill [color=xdxdff] (-1.07453,4.00738) circle (1.5pt);
\fill [color=xdxdff] (-0.47371,2.66161) circle (1.5pt);
\fill [color=xdxdff] (0.9082,3.90551) circle (1.5pt);
\end{tikzpicture}
	\caption{A probability simplex}
\end{figure}

For measurable non-intersecting sets $A$, $B$, $(A\cup B)^c$, we represent the probabilities of $A, B$ and $(A\cup B)^c$ in a probability simplex (as in Figure \ref {fig1}). Here horizontal lines represent measures with constant probability $\mathbb{P}(A)$, and rays extending from $A$ represent measures with constant probability $\mathbb{P}(B|A^c)$. Similarly from the perspective of $B$.

Epstein and Schneider \cite{Epstein2003} show that the risk measure under consideration is $\mathcal{F}^A$-consistent if and only if the set of probability measures forms a `$\mathcal{F}^A$-rectangle', that is, if its projection on this simplex is the quadrilateral enclosed by two horizontal lines (measures of constant $\mathbb{P}(A)$), and two rays extending from $A$ (measures of constant $\mathbb{P}(B|A^c)$). It is then clear that a closed set of measures absolutely continuous with respect to $\mathbb{P}^0$ cannot be both a $\mathcal{F}^A$-rectangle and a $\mathcal{F}^B$-rectangle, unless it lies entirely on one of the sides of the triangle (case (i)), is the whole triangle (case (ii)) or is a single point (case (iii)). 

Our formal argument can be diagramatically represented by taking a measure $Z^1$ with minimal probability for $A$, and measures $Z^2, Z^3$ with differing conditional probabilities $\mathbb{P}(B|A^c)$. Using $\mathcal{F}^A$-consistency, we constructed measures $Z^4, Z^5$, with minimal $\mathbb{P}(A)$ and differing $\mathbb{P}(B|A^c)$. As shown in Figure \ref{fig1}, these measures form the bottom of an $\mathcal{F}^A$-rectangle. By $\mathcal{F}^B$-consistency, we then constructed a measure $Z^6$ with the same $\mathbb{P}(B)$ as $Z^5$ and the same $\mathbb{P}(A|B^c)$ as $Z^4$. By the pasting property, $Z^6\in\mathcal{Z}$, but we then show that $Z^6$ has a lower probability of $A$ than $Z^1$, and so have obtained a contradiction. The required construction to show that $\max_{\mathbb{P}\in\mathcal{P}}\mathbb{P}(A)=1$ is then also clear (construct measures on the top of an $\mathcal{F}^A$-rectangle, then find a measure above them using $\mathcal{F}^B$-consistency).
\end{remark}

Using Lemma \ref{lem:mainlem}, we obtain the following general result.

\begin{theorem}\label{thm:mainthm}
Let $\rho$ be a coherent risk measure having the Fatou property with a time-consistent extension for any (simple) filtration of the probability space $(\Omega, \mathcal{F}, \mathbb{P}^0)$. Then one of the following holds:
\begin{enumerate}
\item \emph{($\rho$ is 1-atomic)} There exists an $\omega_1\in\Omega$ such that $\rho(X) = -X(\omega_1)$.
\item \emph{($\rho$ is 2-atomic)} There exist $\omega_1,\omega_2\in\Omega$ and an interval $[\alpha,\beta]\subsetneq[0,1]$ such that
\[\rho(X) = \max_{p\in[\alpha,\beta]}\{-pX(\omega_1)-(1-p)X(\omega_2)\}.\]
\item \emph{($\rho$ is extremal)} $\rho(X) = \Pesssup\{-X(\omega)\}$.
\item \emph{($\rho$ is linear)} There exists a probability measure $\mathbb{P}^1$ absolutely continuous with respect to $\mathbb{P}^0$ such that $\rho(X) = -E_{\mathbb{P}^1}[X]$, (and $\mathbb{P}^1$ does not assign all its mass to a single atom).
\end{enumerate}
\end{theorem}

\begin{proof}
Let $\mathcal{A}=\{A_i\}\subset\mathcal{F}$ be a finite partition of $\Omega$. We separate this partition into two sets, the polar sets $\mathcal{N}=\{A_i\in\mathcal{A}: \mathbb{P}(A_i)=0 \text{ for all } \mathbb{P}\in\mathcal{P}\}$, and the non-polar sets $\mathcal{B}=\mathcal{A}\setminus \mathcal{N}$. Clearly $\mathcal{B}$ is nonempty.

As $\rho$ is time-consistent for all simple filtrations, it is $\mathcal{F}^{A_i}$-consistent for all $i$. Hence, by Lemma \ref{lem:mainlem}, considering an arbitrary pair of sets in $\mathcal{A}$ we see that we can classify $\mathcal{B}$ as follows:
\begin{enumerate}
\item \emph{($\mathcal{B}$ is 1-atomic)} $\mathcal{B}$ contains one element, (which hence satisfies $\mathbb{P}(B)=1$ for all $\mathbb{P}\in\mathcal{P})$,
\item \emph{($\mathcal{B}$ is extremal)} for all $B_i\in\mathcal{B}$ there exist measures $\mathbb{P}, \mathbb{P}'\in\mathcal{P}$ such that $\mathbb{P}(B_i) = 0$ and $\mathbb{P}'(B_i)=1$,
\item \emph{($\mathcal{B}$ is fixed)} $\mathcal{B}$ contains at least two elements, and for all $B_i\in\mathcal{B}$, $\mathbb{P}(B_i)$ is the same for all $\mathbb{P}\in\mathcal{P}$, or
\item \emph{($\mathcal{B}$ is 2-atomic)} $\mathcal{B}$ contains two elements, and is neither fixed nor extremal.
\end{enumerate}

Now consider a finite refinement $\mathcal{A}'\subset\mathcal{F}$ of $\mathcal{A}$. Without loss of generality, we suppose that $\mathcal{A}'$ is obtained from $\mathcal{A}$ by dividing one set into two. (This process can then be repeated to give any desired finite refinement.)  Clearly $\mathcal{N}'\supseteq \mathcal{N}$. We can then see the following.
\begin{enumerate}
\item If $\mathcal{B}$ is extremal it contains at least two sets. Either $\mathcal{B}'=\mathcal{B}$, or $\mathcal{B}'$ will contain at least one set which is the same as in $\mathcal{B}$. For this set $B_i$, there exist measures $\mathbb{P}, \mathbb{P}'\in\mathcal{P}$ such that $\mathbb{P}(B_i) = 0$ and $\mathbb{P}'(B_i)=1$, and therefore $\mathcal{B}'$ cannot be 1-atomic, 2-atomic or fixed, and so must be extremal.

\item If $\mathcal{B}$ is 2-atomic, then $\mathcal{B}'$ contains at least one set $B_i$ which is in $\mathcal{B}$, for which $\mathbb{P}(B_i)$ may vary in $\mathbb{P}\in\mathcal{P}$, but does not attain both 0 and 1. Hence, as $\mathcal{B}'$ is not 1-atomic, extremal or fixed, $\mathcal{B}'$ is 2-atomic. Note this implies that $\mathcal{B}=\{B_1,B_2\}$ differs from $\mathcal{B}'=\{B_1', B_2'\}$ only by the union with elements of $\mathcal{N}'$, and so $\mathbb{P}(B_j) = \mathbb{P}(B'_j)$, $j=1,2$, for all $\mathbb{P}\in\mathcal{P}$.

\item If $\mathcal{B}$ is fixed, then $\mathcal{B}'$ contains at least two elements and at least one of them is as in $\mathcal{B}$, and so has a fixed probability. Therefore $\mathcal{B}'$ is not 1-atomic, 2-atomic or extremal, and so is fixed.
\end{enumerate}
Therefore, if for a given partition $\mathcal{A}$, the collection of non-polar sets $\mathcal{B}$ is not 1-atomic, then its class (2-atomic, fixed or extremal) remains the same under finite refinement of the partition $\mathcal{A}$.

Now suppose that there exists a finite partition $\mathcal{C}\subset \mathcal{F}$ of $\Omega$ such that the induced $\mathcal{B}$ is not 1-atomic. Otherwise, let $\mathcal{C}=\Omega$, which is clearly 1-atomic. We suppose without loss of generality that all partitions we consider are refinements of $\mathcal{C}$, and hence all refinements will have the same class. We now consider each of the four classes separately.
\begin{enumerate}
\item Suppose $\mathcal{C}$ is 1-atomic. Then let $\mathcal{A}^n$ be a sequence of partitions such that, for all $A_i^n\in\mathcal{A}^n$, either $A_i^n$ is an $\mathcal{F}$-atom, or $\mathbb{P}^0(A^n_i)<2^{-n}$. As $\mathcal{B}^n=\{B^n_1\}$ is 1-atomic, we then see that $B^\infty:=\bigcap_n B^n_1$ is either a $\mathbb{P}^0$-nonnull $\mathcal{F}$-atom, or is a $\mathbb{P}^0$-null set. As $\mathbb{P}(B^\infty)=\lim\mathbb{P}(B^n)=1$ for all $\mathbb{P}\in\mathcal{P}$, $B^\infty$ must then be a $\mathbb{P}^0$-atom. Hence, for any $\omega_1\in B^\infty$, we have that $\rho(X) = -X(\omega_1)$, that is, $\rho$ is 1-atomic.

\item Next suppose $\mathcal{C}$ is 2-atomic. Take the same partition $\mathcal{A}^n$ as in the 1-atomic case, and we can again see that if $\mathcal{B}^n=\{B_1^n, B^n_2\}$, then $B_1^\infty := \bigcap_n B^n_1$ and $B_2^\infty := \bigcap_n B^n_2$ are $\mathbb{P}^0$-nonnull $\mathcal{F}$-atoms. As $\mathbb{P}(B_1^n)=1-\mathbb{P}(B_2^n)$ is independent of $n$, we can define $\alpha = \min_{\mathbb{P}\in\mathcal{P}} \mathbb{P}(B^n_1)>0$ and $\beta=\max_{\mathbb{P}\in\mathcal{P}} \mathbb{P}(B^n_1)<1$. Hence, for any $\omega_1\in B_1^\infty$, $\omega_2\in B_2^\infty$, $\rho(X) = \max_{p\in[\alpha,\beta]}\{-pX(\omega_1) - (1-p) X(\omega_2)\}$, that is, $\rho$ is 2-atomic.

\item Next suppose $\mathcal{C}$ is fixed. Then for any $A\in\mathcal{F}$, we can take a finite partition which is a refinement of both $\mathcal{C}$ and $\{A, A^c\}$, to see that $\mathbb{P}(A)$ is fixed for all $\mathbb{P}\in\mathcal{P}$. Hence $\mathcal{P}$ contains only a single element $\mathbb{P}^1$, and we see that $\rho(X)= -E_{\mathbb{P}^1}[X]$, that is, $\rho$ is linear.

\item Finally, suppose $\mathcal{C}$ is extremal. For any $X\in L^\infty(\Omega, \mathcal{F},\mathbb{P}^0)$, in the same way as we construct the Lebesgue integral, by the dominated convergence theorem we can construct a sequence $\mathcal{A}^n=\{A_i^n\}$ of finite refinements of $\mathcal{C}$ and a sequence of simple approximations of the form $X^n = \sum x_i^n I_{A_i^n}$ such that $X^n\downarrow X$ $\mathbb{P}^0$-almost surely. As $\mathcal{B}^n$ is extremal, \[\rho(X^n)=\max_i\{-X^n(\omega)|\omega\in B^n_i\}=\Pesssup\{-X^n(\omega)\}.\]

As $\rho$ satisfies the Fatou property and $X^n\downarrow X$ a.s., $\rho(X^n)\uparrow\rho(X)$, as shown in \cite{Delbaen2002}. It follows that
\[\begin{split}
\rho(X) &= \sup_n\{\Pesssup\{-X^n(\omega)\}\}\\
&=\Pesssup\{\sup_n\{-X^n(\omega)\}\}\\
&=\Pesssup\{-X(\omega)\},
\end{split}\]
that is, $\rho$ is extremal.
\end{enumerate}
\end{proof}

\begin{corollary}
If $(\Omega, \mathcal{F},\mathbb{P}^0)$ is non-atomic, then $\rho$ is either extremal or linear.
\end{corollary}
\begin{proof}
The 1-atomic and 2-atomic cases correspond to atomic measures $\mathbb{P}\in\mathcal{P}$, but these measures must also be absolutely continuous with respect to $\mathbb{P}^0$, which is a contradiction. 
\end{proof}

\begin{corollary}
If $\rho$ is strictly monotone, that is, when $X\geq Y$ $\mathbb{P}^0$-a.s. and $\mathbb{P}^0(X>Y)>0$ implies $\rho(X)<\rho(Y)$, and if $\mathcal{F}$ does not consist of one or two atoms, then $\rho$ is linear.
\end{corollary}
\begin{proof}
As $\mathcal{F}$ does not consist of one or two atoms, it is clear that the 1-atomic, 2-atomic and extremal cases are not strictly monotone.
\end{proof}

\begin{remark}
If $\mathbb{P}^0\in\mathcal{P}$, then the $\Pesssup$ and $\mathbb{P}^0\text{-ess\,sup}$ agree, and hence an extremal $\rho$ is simply the classical essential supremum.
\end{remark}

Finally, for a general filtration $\{\mathcal{F}_t\}$ of $(\Omega, \mathcal{F})$, we now give a construction of $\{\mathcal{F}_t\}$-consistent extensions of these risk measures in each case. Verification that these are indeed $\{\mathcal{F}_t\}$-consistent extensions of the risk measures is straightforward.
\begin{proposition}
Suppose $\rho$ is 1-atomic. Then let $A$ be the $\mathcal{F}$-atom containing $\omega_1$. Let $B_t=\{\omega: E_{\mathbb{P}^0}[I_A|\mathcal{F}_t]>0\}=\bigcap \left(B\in\mathcal{F}_t: \mathbb{P}^0(B\cap A)>0\right)\in\mathcal{F}_t$.
Then an $\{\mathcal{F}_t\}$-consistent extension of $\rho$ is given by:
\[\rho(X) = \begin{cases} -X(\omega_1) & \omega\in B_t\\ -E_{\mathbb{P}^0}[X|\mathcal{F}_t] & \omega\notin B_t\end{cases}\]
\end{proposition}

\begin{proposition}
Suppose $\rho$ is 2-atomic. Then let $A_1$, $A_2$ be the $\mathcal{F}$-atoms containing $\omega_1$ and $\omega_2$. Let
\[\begin{split}
B_t^1&=\{\omega: E_{\mathbb{P}^0}[I_{A^1}|\mathcal{F}_t]>0\}\\
B_t^2&=\{\omega: E_{\mathbb{P}^0}[I_{A^2}|\mathcal{F}_t]>0\}\\
B_t^3&=(B^1_t\cup B^2_t)^c.
\end{split}\]
If there exists a set $B\in\mathcal{F}_t$ such that $A^1\subseteq B$, $A^2\not\subseteq B$, (i.e. we can distinguish between $\omega_1$ and $\omega_2$ in $\mathcal{F}_t$), then let
\[\rho_t(X) = \begin{cases} 
-X(\omega_1) & \omega\in B_t^1\\
-X(\omega_2) & \omega\in B_t^2\\
-E_{\mathbb{P}^0}[X|\mathcal{F}_t] & \omega\in B_t^3.
\end{cases}\]
Otherwise, let 
\[\rho_t(X) = \begin{cases} 
\rho(X) & \omega\in B_t^1=B_t^2\\
-E_{\mathbb{P}^0}[X|\mathcal{F}_t] & \omega\in B_t^3,
\end{cases}\]
and $\rho_t$ is an $\{\mathcal{F}_t\}$-consistent extension of $\rho$.
\end{proposition}

\begin{proposition}
Suppose $\rho$ is linear, then an $\{\mathcal{F}_t\}$-consistent extension of $\rho$ is given by $\rho_t(X)=-E_{\mathbb{P}^1}[X|\mathcal{F}_t]$, with appropriate modification on any $\mathbb{P}^1$-null sets.
\end{proposition}

\begin{proposition}
Suppose $\rho$ is extremal and $\mathbb{P}^0\in\mathcal{P}$. As shown in \cite{Barron2003}, we can define the $\mathcal{F}_t$-conditional essential supremum, denoted $\mathbb{P}^0\text{-ess\,sup}(\,\cdot\,|{\mathcal{F}_t})$. Then an $\{\mathcal{F}_t\}$-consistent extension of $\rho$ is given by
\[\rho_t(X)=\mathbb{P}^0\text{-ess\,sup}(-X|{\mathcal{F}_t}).\]
\end{proposition}

\section{Conclusion}
We have shown that a coherent risk measure with the Fatou property is time consistent for all (simple) filtrations of a probability space if and only if it is either 1-atomic, 2-atomic, linear or extremal. This result demonstrates the restrictive nature of assuming time-consistency for all filtrations. This result holds whether the filtrations are assumed to be in discrete or continuous time. 

For the more general convex risk measures, the result of Theorem \ref{thm:mainthm} does not hold. For example, $\rho(X) = \gamma \ln E[\exp(-X/\gamma)]$, $\gamma\geq 0$ has a time consistent extension $\rho_t(X)= \gamma \ln E[\exp(-X/\gamma)|\mathcal{F}_t]$. Our result also does not consider the case where there is no reference probability measure $\mathbb{P}^0$. Further work is needed to determine the corresponding result in these cases.

\bibliographystyle{plain}  
\bibliography{../RiskPapers/General}

\begin{thebibliography}{1}

\bibitem{Artzner1999}
Philippe Artzner, Freddy Delbaen, Jean-Marc Eber, and David Heath.
\newblock Coherent measures of risk.
\newblock {\em Mathematical Finance}, 9(3):203--228, 1999.

\bibitem{Artzner2007}
Philippe Artzner, Freddy Delbaen, Jean-Marc Eber, David Heath, and Hyejin Ku.
\newblock Coherent multiperiod risk adjusted values and {B}ellman's principle.
\newblock {\em Annals of Operations Research}, 152:5--22, July 2007.

\bibitem{Barron2003}
E.N. Barron, P.~Cardaliaguet, and R.~Jensen.
\newblock Conditional essential suprema with applications.
\newblock {\em Applied Mathematics and Optimization}, 48(3):229--253, October
  2003.

\bibitem{Delbaen2002}
Freddy Delbaen.
\newblock Coherent risk measures on general probability spaces.
\newblock In {\em Advances in Finance and Stochastics; Essays in Honour of
  Dieter Sondermann}, pages 1--38. Springer, 2002.

\bibitem{Epstein2003}
Larry~G. Epstein and Martin Schneider.
\newblock Recursive multiple-priors.
\newblock {\em Journal of Economic Theory}, 113:1--31, 2003.

\bibitem{Follmer2002}
Hans F\"ollmer and Alexander Schied.
\newblock {\em Stochastic Finance: An introduction in discrete time}.
\newblock Studies in Mathematics 27. de Gruyter, Berlin-New York, 2002.

\end{thebibliography}
\end{document}